\documentclass[11pt]{article}
\usepackage[utf8]{inputenc}
\usepackage{XCharter}                         
\usepackage[margin=1in]{geometry}
\usepackage{fancyhdr}
\usepackage{mypackage}
\usepackage{framed}

\bibliography{bib-refs}

\title{Asymptotically Optimal Hardness for $k$-Set Packing and $k$-Matroid Intersection}
\author{Euiwoong Lee \thanks{University of Michigan 
(\href{mailto:euiwoong@umich.edu}{euiwoong@umich.edu})} 
\and Ola Svensson \thanks{Ecole Polytechnique Fédérale de Lausanne
(\href{mailto:ola.svenssong@epfl.ch}{ola.svenssong@epfl.ch},
\href{mailto:theophile.thiery@epfl.ch}{theophile.thiery@epfl.ch})} 
\and Theophile Thiery \footnotemark[2]}

\definecolor{shadecolor}{named}{WhiteSmoke}

\begin{document}
    \maketitle
    \begin{abstract}
        For any $\e > 0$, we prove that $k$-Dimensional Matching is hard to approximate within a factor of $k/(12 + \e)$ for large $k$ unless $\NP \subseteq \BPP$. Listed in Karp's 21 $\NP$-complete problems, $k$-Dimensional Matching is a benchmark computational complexity problem which we find as a special case of many constrained optimization problems over independence systems including: $k$-Set Packing, $k$-Matroid Intersection, and Matroid $k$-Parity.
        For all the aforementioned problems, the best known lower bound was a $\Omega(k /\log(k))$-hardness by Hazan, Safra, and Schwartz. In contrast, state-of-the-art algorithms achieved an approximation of $O(k)$.
        Our result narrows down this gap to a constant and thus provides a rationale for the observed algorithmic difficulties. The crux of our result hinges on a novel approximation preserving gadget from $R$-degree bounded $k$-CSPs over alphabet size $R$ to $kR$-Dimensional Matching. Along the way, we prove that $R$-degree bounded $k$-CSPs over alphabet size $R$ are hard to approximate within a factor $\Omega_k(R)$ using known randomised sparsification methods for CSPs.
    \end{abstract}

    \section{Introduction}

        

        The $k$-dimensional matching problem consists of finding a maximum collection of disjoint edges in a $k$-partite hypergraph where each edge has size $k$.
        Cited amongst Karp's list of 21 $\NP$-complete problems, it is a benchmark problem for algorithms and approximability results. In particular, it models the maximum bipartite matching problem for $k = 2$ and is a special case of $k$-Set Packing and of $k$-Matroid Intersection.         
        Both problems are central constrained optimisation problems that have received considerable attention over the past years with notable contributions on the algorithmic side as evidenced by: \cite{Hurkens:1989:Size,Halldorson:1995:Approximating,Berman:2000:d/2,Cygan:2013:Improved, Neuwohner:2021:Improved, Neuwohner:2023:Passing, Thiery:2023:Improved, Lee:2013:Matroid-Matching-journal, Linhares:2020:Multi-matroid-intersection}.
        For these problems, the state-of-the-art approximation ratios are of the form $O(k)$. Cygan \cite{Cygan:2013:Improved} designed a $\frac{k+1}{3}$-approximation algorithm for $k$-Set Packing while Lee, Sviridenko and Vondr\'ak \cite{Lee:2013:Matroid-Matching-journal} obtained a $\frac{k}{2}$-approximation algorithm for $k$-Matroid Intersection.
        In contrast, on the hardness front, the best lower bound for \emph{all} these problems remains $\Omega(k/\log(k))$ by Hazan et al. \cite{Hazan:2006:Complexity} who proved that $k$-Dimensional matching is $\NP$-hard to approximate within the same ratio. For small values of $k$, Berman and Karpinski \cite{DBLP:journals/eccc/ECCC-TR03-008} showed it is $\NP$-hard to approximate 
        $k$-Dimensional Matching beyond a factor $98/97, 54/53, 30/29$ for $k = 3, 4, 5, 6$ respectively. 
        Our main result is the following:
        \begin{shaded}
        \begin{restatable}{theorem}{Main}
        \label{thm:main-matching}
            Unless $\mathbf{NP} \subseteq \mathbf{BPP}$, for any constant $\eps > 0$ and sufficiently large $k \geq k_0(\eps)$, there is no polynomial-time algorithm that approximates $k$-Dimensional Matching within a factor of $k/(12 + \eps)$.
        \end{restatable}
        \end{shaded}

        In particular, it explains the lack of substantial algorithmic progress beyond $O(k)$-approximation in that any algorithm is tied to an approximation ratio of that form.
        Apart from $k$-Set Packing and $k$-Matroid Intersection, $k$-Dimensional Matching is a reference problem whose hardness carries over to further generalizations of that problem. A non-exhaustive list of these generalizations includes: Independent Set in $k+1$-Claw Free Graph, $k$-Matchoid, and $k$-Matroid Parity (see: \cite{Thiery:2023:Improved, Lee:2013:Matroid-Matching-journal} for definitions and comparisons between these problems). All admit $O(k)$-approximation algorithms while the best $\NP$-hardness bound is equal to $\Omega(k/\log(k))$ from Hazan et al.'s result \cite{Hazan:2006:Complexity}, except for the independent set problem in $k+1$-claw free graph whose hardness was improved to $\frac{k+1}{4}$ \cite{Lee:2024:Hardness, Minzer:2024:Near-Optimal}. \Cref{thm:main-matching} thus reduces the gap between approximability and hardness from $O(\log(k))$ to a constant (for large $k$). A hierarchy of the different problems discussed so far can be found in  \Cref{fig:enter-label}.
        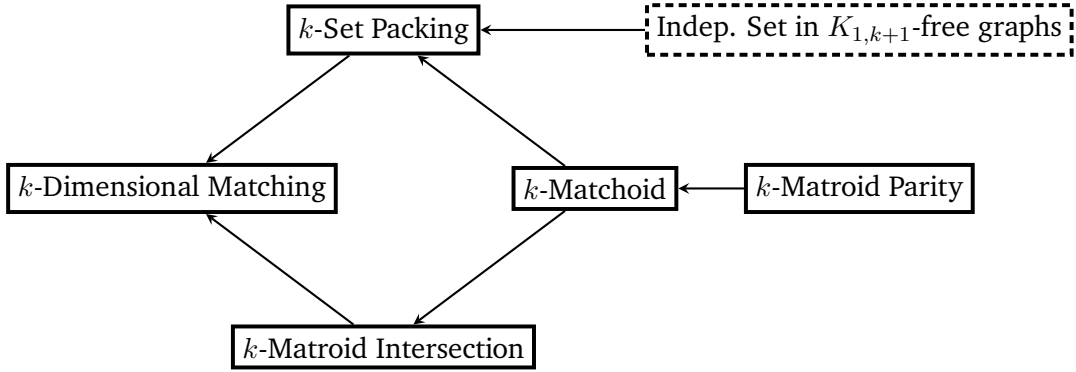
\begin{figure}[h]
            \centering
            \begin{tikzpicture}[x=1pt,y=1pt,yscale=-1,xscale=1]

\node[draw, fill=white, line width=1.5, align=center] (kDM) at (80, 140) {$k$-Dimensional Matching};
\node[draw, fill=white, line width=1.5, align=center] (kSetPacking) at (160, 80) {$k$-Set Packing};
\node[draw, fill=white, line width=1.5, align=center] (kMatroidIntersection) at (160, 200) {$k$-Matroid Intersection};
\node[draw, fill=white, line width=1.5, align=center] (kMatchoid) at (240, 140) {$k$-Matchoid};
\node[draw, fill=white, line width=1.5, densely dashed, align=center] (ISClawFree) at (340, 80) {Indep. Set in $K_{1, k+1}$-free graphs};
\node[draw, fill=white, line width=1.5, align=center] (kMatroidParity) at (340, 140) {$k$-Matroid Parity};

\draw[-stealth, line width=0.8pt] (kSetPacking) -- (kDM);  
\draw[-stealth, line width=0.8pt] (kMatroidIntersection) -- (kDM);  
\draw[-stealth, line width=0.8pt] (kMatchoid) -- (kSetPacking);  
\draw[-stealth, line width=0.8pt] (kMatchoid) -- (kMatroidIntersection);  
\draw[-stealth, line width=0.8pt] (kMatroidParity) -- (kMatchoid);  
\draw[-stealth, line width=0.8pt] (ISClawFree) -- (kSetPacking);  

\end{tikzpicture}
            \caption{This diagram represents a hierarchy of problems that capture $k$-Dimensional Matching. An arrow from $P$ to $Q$ means that $Q$ can be cast as $P$. For all problems with solid boxes \Cref{thm:main-matching} improves the hardness bound from $\Omega(k/\log(k))$ to $k/12$. On the other hand, finding an independent set in a $k+1$-claw free graph is hard to approximate beyond a factor of $\frac{k+1}{4}$ \cite{Lee:2024:Hardness, Minzer:2024:Near-Optimal}.}
            \label{fig:enter-label}
        \end{figure}

    \subsection{From CSPs to $k$-Dimensional Matching}
    The approximability of $k$-Dimensional Matching is related to that of $k$-CSPs, where we assign labels to variables to maximally satisfy constraints involving $k$ variables.
    Hazan et al. prove $\NP$-hardness of $k$-Dimensional Matching by providing a reduction from  3-LIN($q$) to $k$-DM.
    Beyond 3-LIN($q$) the approximability of CSPs was studied in parallel and led to strong results subject to various parameter restrictions \cite{Hastad:2000:bounded, Trevisan:2001:Non-Approximability,Laekhanukit:2014:Parameters, Lee:2024:Hardness}.
    The parameters that we will be interested in this work are the degree $d$ (number of constraints involving a variable) and the maximum number of labels $R$ (alphabet size). The best $\NP$-hardness results in terms of $R$ and $d$ are $O(R^{-(k-2)})$ for $k$-CSPs \cite{Chan:2016:Approximation-Resistance}, and $d/2$ for $2$-CSPs \cite{Minzer:2024:Near-Optimal}.
    These results combined with clever reductions led to stronger inapproximability bounds for connectivity problems in graphs \cite{Laekhanukit:2014:Parameters, Manurangsi:2019:Note-degree, Minzer:2024:Near-Optimal} and for finding independent sets in $d$-claw-free graphs \cite{Lee:2024:Hardness}.
    Motivated by these advances, we prove a new hardness result for $R$-degree bounded $k$-CSPs with alphabet $R$. Our hardness proof closely follows that of \cite{Lee:2024:Hardness}: We start from a $d$-regular $k$-CSP instance over alphabet size $R$ and randomly sample the constraints to obtain a $R$-degree bounded $k$-CSP hard to approximate within a factor $\Omega(R)$. We give a more detailed description of this procedure in the next section.
    The crux of our result then lies on a new approximation-preserving reduction from $R$-degree bounded $k$-CSP with alphabet size $R$ to $kR$-Dimensional Matching, which immediately implies that $p$-DM is hard to approximate beyond a factor $\Omega(p)$.

    \paragraph{Further References:} The hardness of $k$-CSPs over alphabet size $R$ is also well understood with a rich line of work \cite{
    Khot:2007:Optimal,
    samorodnitsky2006gowers,
    Guruswami:2008:Constraint,
    Austrin:2009:Approximation,
    Chan:2016:Approximation-Resistance,
    manurangsi2015near}.
    Hardness of approximation of factor $O(k/R^{k-2})$ for every $k, R$ and that of factor $O(k/R^{k-1})$ for every $k \geq R$ was first proved under the Unique Games Conjecture (UGC)~\cite{samorodnitsky2006gowers, Austrin:2009:Approximation} and later just under $\mathbf{P} \neq \mathbf{NP}$~\cite{Chan:2016:Approximation-Resistance}.
    This result is tight when $k > R$ \cite{DBLP:journals/toc/MakarychevM14}. When $k \leq R$, the hardness was subsequently improved assuming the UGC \cite{manurangsi2015near, lee2022characterization}, albeit without almost perfect completeness. For $2$-CSPs under the assumption that the instance is $d$-degree bounded, the optimal hardness of approximation of factor $d/2$ was first assuming the UGC \cite{Lee:2024:Hardness} and later without it \cite{Minzer:2024:Near-Optimal}.
    
    \subsection{Technical Overview}
    We start by informally discussing the reduction from some \emph{hard} $k$-CSP instance $\Pi$ to a hypergraph whose matchings should correspond to satisfying assignments of $\Pi$. Our high-level strategy follows that of the previous best $\Omega(k / \log (k))$-hardness of Hazan, Safra, and Schwartz~\cite{Hazan:2006:Complexity}. Let $\Pi = (G = (V, E), R, \cC)$ be a $d$-regular $k$-CSP with alphabet size $R$ (see \Cref{def:CSP}) for some $d$ that might be arbitrarily larger than $R$.
    We construct a \emph{variable gadget} for each $v \in V$, i.e. an hypergraph $H_v = (X_v, E_v)$ with the following properties:
    \begin{enumerate}
        \item $E_v$ is partitioned into $R$ matchings $E_1, \dots, E_R$ with $|E_1| = \dots |E_R| = d$, where $E_i = \{ e_{i, j} \}_{j \in [d]}$.
        \item Any matching must be almost contained in some $E_i$; formally, for every matching $M \subseteq E$, there exists $i$ such that $|M \setminus E_i| \leq \delta \card{E_i}$ for some small $\delta \geq 0$.
    \end{enumerate}
    The intuition is that any matching in $H_v$ is almost contained in some set $E_i$ and so should correspond to the situation where $v$ is assigned label $i$ in $\Pi$. The $\nth{j}$ edge in $E_i$ corresponds to the $\nth{j}$ constraint involving $v$ for $j \in [d]$ where each $v$ arbitrarily orders its constraints.
    The final hypergraph is constructed as follows: the vertex set is the union of the vertex sets over all gadgets. Let $C \in \cC$ be a clause involving $k$ variables $v_1, \ldots, v_k \in V$, assume that $C$ is the $\nth{b_i}$ constraint in $v_i$'s ordering (so $b_i \in [d]$ for each $i \in [k]$). For every satisfying assignment $a = (a_1, \ldots, a_k) \in [R]^k$ of $C$, we create an edge $e(C, a) = e^{v_1}_{a_1, b_1} \cup
    \ldots \cup
    e^{v_k}_{a_k, b_k}$, where the superscript indicates the gadget.
    Crucially, if there is a good assignment $\alpha : V \to [R]$ for $\Pi$ that satisfies $\cC' \subseteq \cC$, then one can check that
    there is a matching $M$ of size $\card{\cC'} = \card{M}$. Indeed, the set
    $M = \{ e(C, \alpha(C)) : C \in \cC' \}$\footnote{Let $\alpha(C) := (\alpha(v))_{v \in e}$ where $e$ is the hyperedge corresponding to $C$.}
    forms a matching with $|\cC'| = |M|$, because for each gadget for vertex $v$, we only use edges in $E_{\alpha(v)}$, which implies that there is no conflict inside $v$'s gadget.
    The other direction, thanks to Property 2 above, also approximately holds so that a large matching implies a good CSP assignment, completing the reduction.
    The final hardness ratio of the reduction is $c/(s + k\delta)$, where $c, s, \delta \in [0, 1]$ are the following parameters.
    \begin{enumerate}[i.]
        \item $\delta$ in Property 2, which denotes the fraction of hyperedges one can get by {\em cheating} compared to the intended matching. 
        \item The starting $k$-CSP instance $\Pi$ is $(c, s)$-hard, meaning that it is hard to distinguish whether the maximum fraction of satisfied constraints is at least $c$ or at most $s$. 
    \end{enumerate}
    \paragraph{Hazan, Safra and Schwartz:}
        Hazan et al.
        \cite{Hazan:2006:Complexity} used the hardness of 3-LIN($R$)~\cite{Hastad:2001:Some-optimal} that has $c = 1 - \eps$, $s = (1+\eps)/R$.
        They constructed a gadget with $\delta = 1/R$ which yields the gap of $R/(4+\eps)$ for Set Packing, but the uniformity of their gadget hypergraph is $\Theta(R \log R)$ so that their hardness for $K$-Set Packing is $\Omega(K / \log K)$. Their result extends to $K$-Dimensional Matching.
        While the gap between $c$ and $s$ can be easily made larger by going from $3$-CSPs to $4$-CSPs, their tradeoff between $\delta$ and the uniformity is optimal in some sense;
        using~\cite{Radhakrishnan:2000:Bounds-for-dispersers}, they proved whenever $d \gg R$ the edge uniformity must be at least $\Omega(R \log(R))$ for $\delta = 1/R$ to hold. 

    \paragraph{Bounded degree yields better gadget.}
        Therefore, one can possibly design a better gadget by getting a hold of $d$ in terms of $R$, even with $\delta = 0$.
        It turns out that the following simple construction gives a \emph{variable gadget} with these guarantees.
        Let $R = d$ be a prime number and consider the hypergraph $(X, E)$ where $X = \mathbb{F}_R^2$ and $E = \{ e_{a, b} : a, b \in \mathbb{F}_R \}$ with $e_{a, b} = \{ (x, ax + b) : x \in \mathbb{F} \}$. Letting $E_a = \{ e_{a, b} \}_{b \in \mathbb{F}_R}$. It is easy to see that $(X, E)$ satisfies both Property 1 and 2 for the gadget with the uniformity $R$ (instead of $\Omega(R \log R)$) and $\delta = 0$! Using this gadget to construct our hypergraph we would obtain an approximation preserving reduction that maps $\Pi$ to an hypergraph $G_\Pi$ such that the maximum matching yields an optimal assignment for $\Pi$. Yet, this gadget assumes that each $d = R$ so each variable appears in at most $R$ constraints ($R$ is also the alphabet size).

        \paragraph{Obtaining bounded degree hardness.}
            Our final hardness is then only determined by Factor ii. It is the hardness for an $R$-degree bounded $k$-CSP with alphabet size $R$ and $k = O(1)$. Our final result is a $R(k-3)/(k(1 + \eps))$-hardness for $R$-degree bounded $k$-CSP with alphabet size $R$ (\Cref{thm:bounded-degree-hardness-simple}).
            The proof closely follows techniques of Lee and Manurangsi~\cite{Lee:2024:Hardness} that prove $d/(2+\eps)$-hardness for $d$-degree bounded $2$-CSP (without restriction on the alphabet size). The strategy is simple.
            We start from a $(1 - \delta, O(R^{-(k-2)}))$-hard $d$-regular (where $d$ can be arbitrarily larger than $R$) $k$-$\CSP$ instance $\Pi$ with alphabet size $R$ \cite{Chan:2016:Approximation-Resistance}.  
            We obtain a hard-to-approximate $R$-regular $k$-$\CSP$ instance $\Pi'$ with alphabet size $R$ by sampling each constraint of $\Pi$ with probability roughly $\simeq R/d$. We ensure that the degree of every vertex is at most $R$ using few deletions which we show have negligible impact.
            The main technical step is to bound the soundness of $\Pi'$. For simplicity, let $s = 1/R^{k-2}$ be the soundness of $\Pi$.
            Let $n$ and $m_0 = nd/k$ be the number of vertices and hyperedges in $\Pi$, and let $m \simeq nR/k$ be the (expected) number of edges in $\Pi'$. The expected number of satisfied constraints after sampling is: $ \mu = s m$. We prove that the soundness of $\Pi'$ is at most $s' = k(1+\e)/((k-3)R)$ for some small $\e > 0$ by showing that, for any assignment, the probability that it satisfies more than $s'm$ constraints is at most: $\lb \frac{s}{s'} \rb^{\mu \cdot \frac{s'}{s}} \simeq R^{-(k-3)s'm} \simeq R^{-(1+\e)n}$.
            We conclude that there is no assignment that satisfies more than $s'm$ constraints using a union bound over all $R^n$ assignments.

    \newpage
    \section{Preliminaries}
    \label{sec:def}
    In this section we introduce basic definitions about $k$-CSPs and $k$-Dimensional Matching.
    \begin{definition}[$k$-CSP]
       \label{def:CSP}
        Given $k \in \NN$, a $k$-CSP instance $\Pi = (G = (V, E), R, \cC)$ consists of:
        \begin{itemize}
            \item A constraint hypergraph $G = (V, E)$ with hyperedges of size $k$.
            \item An alphabet $[R]$.
            \item For each $ e = (u_1, \ldots, u_k) \in E$, a constraint $\cC_e \subseteq R^k$. We denote by $\cC$ the set of constraints, and $ \card{\cC} = \card{E}$.
        \end{itemize}
    \end{definition}
    The graph terminology applies to describe $k$-CSP instances. We say that $\Pi = (G = (V, E), R, \cC)$ is \emph{$d$-degree bounded} (respectively \emph{$d$-regular}) if every vertex has degree at most at $d$ (respectively \emph{exactly} $d$). We say that $\Pi$ is $k$-partite if $G$ is a $k$-partite graph (i.e. $V = V_1 \cup \ldots V_k$ with $V_i \cap V_j = \emptyset$ and each $e \in E$ is incident to exactly $1$ vertex from each $V_i$).
    An \emph{assignment} is a tuple $(\psi_{v})_{v \in V}$ such that $\psi_v \in [R]$. In other words, it is an assignment of a label to each vertex $v \in V$, denoted by $\psi_v$.
    We are interested in the number of constraints satisfied $\psi$, and we denote by $\cC(\psi)$ the set of constraints satisfied by the assignment $\psi$.
    More precisely, we define $\cC(\psi) \triangleq \lc e \in E \colon \psi(e) \in \cC_e \rc$.
    Let $\val_{\Pi}(\psi) \triangleq \card{\cC(\psi)} / |\cC|$ be the fraction of the constraints satisfied by $\psi$. The maximum fraction of constraints satisfied by any assignment is denoted by $\val(\Pi)$ and we denote by $\psi^*$ an assignment that realizes $\val_{\Pi}(\psi^\ast) = \val(\Pi)$.
    Given a $k$-CSP instance $\Pi$, we say that $\Pi$ is $(c, s)$-hard if it is $\NP$-hard to distinguish whether $\val(\Pi) \geq c$ or  $\val(\Pi) \leq s$.
    \begin{remark}
        \label{assumption}
        Without loss of generality, all the $k$-CSPs that we mention in this work are $k$-partite. 
    \end{remark}

    \begin{definition}[$k$-Set Packing/$k$-Dimensional Matching]
        \label{def:k-SP}
         A $k$-Set Packing instance $\Pi = (G = (V, E))$ consists of: an hypergraph $G = (V, E)$ with hyperedges of size at most $k$. We say that $\Pi$ is a $k$-Dimensional Matching instance if $G$ is $k$-partite.
    \end{definition}
    Note that in the special case of $k$-Dimensional Matching every hyperedge has size exactly $k$.    We will be interested in the matching of maximum size in $G$. A matching $M \subseteq E$ is a subset of edges where any vertex belongs to at most one edge in $M$.

    \section{Approximation Preserving Reduction from $k$-CSP to $kR$-Set Packing}
    \label{sec:reduction}
    This section details our main gadget. It is an approximation-preserving reduction from $R$-degree bounded $k$-CSP with alphabet size $R$ to $kR$-Set Packing.
    \begin{theorem}
        \label{thm:reduction-to-SP}
        Let $R \in \NN$ be a prime number.
        There is an approximation-preserving reduction that maps any $R$-degree bounded $k$-$\CSP$ instance $\Pi = (G = (V, E), R, \cC)$ with alphabet size $R$ with optimal assignment $\psi^\ast$ to a $k R$-Set Packing instance with maximum matching $M^\ast$ such that $\card{\cC(\psi^\ast)} =  \card{M^\ast}$.
        If $\Pi$ is $k$-partite, the constructed instance is a $kR$-Dimensional Matching instance.
        The running time of the reduction is at most $\poly(|V|, |E|, R^k)$. 
    \end{theorem}
    \begin{remark}
        In fact, our reduction does not only preserve size. A maximum matching can be used to find an optimal assignment of $\Pi$ and vice-versa.
    \end{remark}
    \begin{proof}[Proof of \Cref{thm:reduction-to-SP}]
       We start with the construction of a gadget that we will later use to construct our reduction.
       Fix a variable $v \in V$ from our $\CSP$ with degree $d_v \in [R]$. We construct a gadget graph $H_v = (X, E_v)$, where $X = [R] \times [R]$. Similar to \cite{Hazan:2006:Complexity}, the idea is to create an edge $e(v, C, a_v) \in E_v$ for each constraint $C \in \cC$ where $v \in C$ and any assignment $a_v \in [R]$ of $v$.
       Since the $\CSP$ is $R$-degree bounded and the alphabet size is equal to $R$, we construct exactly $d_v R \leq R^2$ edges.
       To construct them, we define the following functions: $f_{a, b}(x) \triangleq a x + b\mod R$, with $x, a, b \in [R]$. We interchangeably let $R$ and $0\!\mod R$ as the same value. 
       For a fixed $a, b \in [R]$, we define an edge as $e_{a, b} = \bigcup_{x \in R} \{(x, f_{a, b}(x)) \}$, which we think of as the plot of an affine function in the $R\times R$ square with coefficients $a, b$.
       Fix an arbitrary one-to-one correspondence $b'$ from the constraints containing $v$ to $[d_v]$.
       Then $e(v, C, a_v) := e_{a_v, b'(C)}$.
       Let $E_a \triangleq \bigcup_{b \in [d_v]} \{ e_{a, b} \}$.
       We will treat $E_a$ as the set of edges that assign $a_v = a$ and alternatively think of $e_{a, b}$ as having color $a$.
       Therefore, each $e_{a, b}$ corresponds to an assignment $a_v = a$ and the $\nth{b}$-constraint where $v$ occurs. The following claim proves a key property of our gadget:
       \begin{claim}
            \label{claim:intersection}
            Distinct edges of the same color do not intersect. Edges of different colors intersect.
       \end{claim}
        \Cref{claim:intersection} implies that any matching in $H_v$ is \emph{consistent}: all edges of a given matching are colored with a unique color corresponding to an assignment $a_v \in [R]$ of $v$. The size of any matching in $H_v$ is bounded by $d_v$.
       \begin{proof}[Proof of \Cref{claim:intersection}]
        Let $e_{a, b}, e_{a, b'}$ be distinct edges of the same color. They intersect if and only if there is an $x \in [R]$ such that $a x + b \equiv a x + b' \mod R$. This would imply that $b \equiv b' \mod R$, a contradiction.
        Similarly, $e_{a, b}, e_{a', b'}$ be two edges of different colors, i.e. $a \neq a'$. We verify the existence of an $x$ such that $a x + b \equiv a' x + b' \mod R$, which we can equivalently write as $(a -a')  x \equiv b'- b  \mod R$. Given that $R$ is a prime number, we can let $x \equiv (b' - b)(a - a')^{-1}$ where $(a - a')^{-1}$ denotes the inverse of $a - a'$ in the field $\mathbb{Z}/R\mathbb{Z}$. We emphasize that the existence of an inverse follows from the fact that $\mathbb{Z}/R\mathbb{Z}$ is a field since $R$ is a prime number.
       \end{proof}

       \subparagraph{Final Construction:} We are now ready to construct our final $k R$-Set Packing instance $G_{\Pi} = (V_{\Pi}, E_{\Pi})$. For each variable $v$ in our $k$-$\CSP$, we construct a \emph{variable} graph $H_v = (X_v, E_v)$ as previously. The ground set $V_{\Pi}$ is the union of each gadget $V_{\Pi} = \bigcup_{v \in V} X_v$.
       Now, each edge in $e \in E_{\Pi}$ will correspond to a constraint $C$ and a satisfying assignment of that constraint. More precisely, for a constraint $C$ associated to $v_1, \ldots, v_k$ we create an edge $e(C, a) \triangleq  e(v_1, C, a_{v_1}) \cup e(v_2, C, a_{v_2})\cup \ldots \cup e(v_k, C, a_{v_k})$, where $e(v_i, C, a_{v_i}) \in E_{v_i}$  if and only if the assignment $a_{v_i}$ of $v_i$'s satisfies the constraint $C$.
       Note that the running time of the reduction and the number of sets our instance is $
       \poly(|V|, |E|, R) \cdot \sum_{C \in \mathcal{C}} [\mbox{number of satisfying assignments for }C]$, which is at most $\poly(|V|, |E|, R^k)$. 
       
        \subparagraph{$k$-Partiteness implies $kR$-DM.}
            Suppose that $\Pi$ is $k$-partite, so that $V = V_1 \cup \ldots \cup V_k$ and each $e \in E$ contains exactly one vertex from each $V_i$. We can write each vertex $u \in V_\Pi$ as $u = (v, (x, j))$ where $v \in V_i$ is a variable in one of the partition of $\Pi$, and $(x, j) \in [R] \times [R]$ is a pair of indices of the variable gadget $X_v$, where $X_v$ can be partitioned according to the column indexed by $x \in [R]$.
            This defines the following partition of $V_{\Pi}$ into $(V^\Pi_{i, x})_{i \in [k], x \in [R]}$ where
            $V^\Pi_{i, x} \triangleq \cup_{v \in V_i} (X_v \cap (\{ x \} \times [R]))$.
            For any vertex $v \in V_i$ of $G$ and an edge $e_{a, b} \in E_v$, the definition of $e_{a, b}$ ensures that $|e_{a, b} \cap V^\Pi_{i, x}| = 1$ for every $x \in [R]$. Since every edge $e(C, a) \in E_{\Pi}$ is the union of $k$ such $e_{a, b}$'s coming from each of $V_1, \dots, V_k$, it has exactly one vertex from every $V^\Pi_{i, x}$.

       \subparagraph{Equivalence.}
       We finish the proof of \Cref{thm:reduction-to-SP} by showing that the size of a maximum matching in $G_{\Pi}$ is equal to the maximum number of simultaneously satisfied constraints in $\Pi$. Let $\psi^\ast$ be an optimal assignment and $M^\ast$ be an optimal matching on $G_{\Pi}$.
       By \Cref{claim:intersection}, there is a one-to-one correspondence between edges of $M^\ast$ and satisfied constraints by $\psi^\ast$.
       Indeed, the matching $M^\ast$ corresponds to a unique assignment $a_v \in [R]$ to each $v \in V$ and thus $\card{M^\ast} \leq \card{\cC(\psi^\ast)}$. On the other hand, the assignment $\psi^\ast$ can be turned into a matching $M$ where an edge belongs to $M$ if the corresponding constraint is satisfied. \Cref{claim:intersection} asserts that this is indeed a matching. Thus, we have $\card{\cC(\psi^\ast)} = \card{M} \leq \card{M^\ast}$. This finishes the proof.
    \end{proof}

    \section{From $k$-$\CSP$ to bounded degree $k$-$\CSP$}
    \label{sec:degree-reduction}
    In this section, we show the hardness of $R$-degree bounded $k$-CSP with alphabet size $R$, proving the following theorem.
    \begin{theorem}
    \label{thm:bounded-degree-hardness-simple}
    Let $k \geq 4$ be an integer. Unless $\mathbf{NP} \subseteq \mathbf{BPP}$, for any $\eps > 0$ and sufficiently large prime
    $R \geq R_0(\eps, k)$, no polynomial-time algorithm can distinguish that a given $R$-degree bounded $k$-CSP instance $\Pi$ with alphabet size $R$ has $\val(\Pi) \geq 1-\eps$ or
    $\val(\Pi) \leq \frac{k(1 + \e)}{(k-3) R}$.
    \label{thm:4-CSP-bounded-hardness}
    \end{theorem}
    Since it implies that $R$-degree bounded $6$-CSPs with alphabet size $R$ are hard to approximate within a factor of $R/(2(1+\e))$ for any $\eps > 0$, the approximation-preserving reduction to $6R$-Dimensional Matching (\Cref{sec:reduction}) implies that $K$-Dimensional Matching is hard to approximate within a factor $K/(12 + \e)$ for any $\e > 0$ and large number $K \geq K_\e$ thereby proving Theorem~\ref{thm:main-matching}. This, in fact, proves \Cref{thm:main-matching} only when $K = 6R$ and $R$ is a prime. For clarity and completess, we show in \Cref{sec:simple-hardness-computation} that the result holds for all sufficiently large $K$.
    For the rest of the section, we prove Theorem~\ref{thm:4-CSP-bounded-hardness}. Our starting point is the following result of Chan~\cite{Chan:2016:Approximation-Resistance}.

    \begin{theorem}[\cite{Chan:2016:Approximation-Resistance}]
        \label{thm:4-CSP-Hardness}
        Let $k \geq 3$. For any $\e > 0$ and prime power $R$, there is a $(1-\e, O(R^{-(k-2)}))$-hard $d$-regular $k$-CSP instance $\Pi$ over alphabet size $R$.
    \end{theorem}

    Next, we prove our main degree-reduction theorem which implies \Cref{thm:bounded-degree-hardness-simple} as a corollary.


        \begin{theorem}
            \label{thm:degree-reduction-general}
            Let $\lambda \in (0, 1)$, $C \in (0, \infty)$, $k \in \NN$ and let $R \in \NN$ be a sufficiently large number.
            Given a $d$-regular $k$-CSP instance $\Pi$ over alphabet size $R$, there is a randomized polynomial-time reduction from $\Pi$ to a $R$-degree bounded $k$-CSP instance $\Pi'$ with alphabet size $R$ such that with high probability the following holds:
            \begin{itemize}
                \item (Completeness) $\val(\Pi') \geq \val(\Pi) - 3\lambda$,
                \item (Soundness) If $\val(\Pi) \leq C R^{-\gamma}$ for some $\gamma \geq 2$, then $\val(\Pi') \leq \frac{k(1 + \lambda)}{(\gamma - 1)(1-\lambda)^2R}$.
            \end{itemize}
        \end{theorem}
        Thus, we can transform a $d$-regular $k$-CSP instance over alphabet size $R$ into a $R$-degree bounded $k$-CSP while ensuring completeness and increasing the soundness by a factor $\simeq k R^{\gamma-1}/(\gamma - 1)$. The proof of \Cref{thm:degree-reduction-general} follows closely that of \cite{Lee:2024:Hardness}.



        \begin{proof}[Proof of Theorem~\ref{thm:degree-reduction-general}]
            Let $R$ be a sufficiently large number such that $R \geq R_0$ where we define $R_0 \triangleq \max \{ k \cdot 100 \lambda^{-3}, 
            \lb \frac{eC(1-\lambda)}{k(1+\lambda)}\rb^{1/\lambda}, 
            100^{1/\lambda}, 
            \frac{C(\gamma - 1)(1- \lambda^2)}{k(1+\lambda)}\}$.
            This property is helpful to ensure that our future computations hold with high probability.
            The statement of the theorem is trivial if $d \leq R$. Thus, we assume throughout the rest of the proof that $d \geq R$.
            Let $\Pi \triangleq ( G = (V, E), R, \cC )$ be a $d$-regular $k$-CSP instance with alphabet size $R$.
            Since $G$ is $k$-partite, we further denote $V = U_1 \cup \ldots \cup U_k$ as the $k$-way partition of the vertex set.
            We construct $\Pi'$ by independently sampling each constraint with probability $p \triangleq (1 - \lambda)R/d$ and deleting a few arbitrary edges. More precisely:
        \begin{itemize}
            \item Let $G_0 \triangleq G$. For each $e \in E$, discard $e$ with probability $1-p$ and denote by $G_1 = (V, E_1)$ the remaining graph.
            \item For each $v \in V$, such that $\deg_{G_1}(v) > R$, remove $R - \deg_{G_1}(v)$ arbitrary edges incident to $v$. Let $G_2 = (V, E_2)$ be the remaining graph. The final $\CSP$ is $\Pi' = (G_2 = (V, E_2), R, \cC_{\mid E_2})$.
        \end{itemize}
        Clearly, the CSP $\Pi'$ is $R$-degree bounded. Note that instead of sampling each edge with probability $R/d$, we sample them with probability $p \triangleq (1 - \lambda)R/d$ for some small $\lambda > 0$. This will be helpful to bound the number of deleted edges.
        Let $n = \card{V}$ . Our initial graph $G_0$ has $\card{E} = \card{U_1} d  = n/k \cdot d$ edges.
        After sampling, the expected number of edges is equal to $\expect{\ld \card{E_1} \rd} = p \card{E} = (1- \lambda) \card{U_1} R = (1- \lambda) m $, where $m \triangleq \card{U_1} R = \frac{n}{k} R$. We think of $m$ as the expected number of edges in our final graph (if $\lambda = 0$). The following $3$ claims (proved in the appendix) are helpful for the rest of the proof.

        \begin{restatable}{claim}{ClaimA}
        \label{thm:ClaimA}
         Suppose that $R \geq R_0$. Given any assignment $\psi$, we let $\card{E_1(\psi)}$ be the number of satisfied constraints in $G_1$ by $\psi$. Let $\cE_1$ be the event that "$\card{E_1(\psi)} \geq m \lb \val_{\Pi}(\psi) - 2\lambda \rb$".
         Then, $\prob{\cE_1} \geq 0.99$.
        \end{restatable}
        \begin{restatable}{claim}{ClaimB}
            \label{thm:ClaimB}
                Suppose that $R \geq R_0$ and let $\cE_2$ be the event "$\card{E_1 \bb E_2} \leq \lambda m$", which corresponds to the event where few deletions occur. Then $\prob{\cE_2} \geq 0.99 $.
        \end{restatable}
        \begin{restatable}{claim}{ClaimC}
            \label{thm:ClaimC}
            Suppose that $R \geq R_0$ and let $\cE_3$ be the event "$\card{E_1} \in [(1 - 2\lambda)m, m]$". Then, $\prob{\cE_3} \geq 0.99$.
        \end{restatable}

        \subparagraph*{Completeness:} We prove that $\val(\Pi') \geq \val(\Pi) - 3\lambda$ for some arbitrarily small $\lambda > 0 $.
        To prove this statement, we use that the fraction of constraints satisfied by $\psi^\ast$ in $G_1$ is still close to its expectation and that very few edges are been deleted.
        Therefore, condition on $\cE_1 \wedge \cE_2 \wedge \cE_3$ that holds with probability $\prob{\cE_1 \wedge \cE_2 \wedge \cE_3} \geq 1 - \sum_{i = 1}^3 \prob{\bar{\cE}_i} \geq 0.97$ by \Cref{thm:ClaimA}, \Cref{thm:ClaimB}, and \Cref{thm:ClaimC}, we have
        \begin{align*}
           \val(\Pi') & \geq \val_{\Pi'}(\psi^\ast) \geq \frac{\card{E_1(\psi^\ast)} - \card{E_1 \bb E_2}}{\card{E_2}} \geq \frac{\card{E_1(\psi^\ast)} - \card{E_1 \bb E_2}}{\card{E_1}} \geq  \val(\Pi) - 3\lambda.
        \end{align*}

        \subparagraph*{Soundness:} Suppose now that $\Pi$ is such that $\val(\Pi) \leq C R^{-\gamma}$ for some constant $C$ and $\gamma \geq 2$. Let $s \triangleq C R^{-\gamma}$ be the \emph{starting soundness}, and let $s' \triangleq \frac{C'}{R}$ be the \emph{target} soundness where $C' = \frac{k(1 +\lambda)}{(\gamma - 1) (1 - \lambda)}$.
        As eluded before, our proof works as follows: we denote by $\cE_{\psi}$ the event where $E_1(\psi)$ has soundness at most $s'$. That is $\cE_{\psi}$ is the event "$\card{E_1(\psi)} \leq s' m$".
        For any $\psi$, we have that $\mu = \expect{\ld \card{E_1(\psi)} \rd} = p \card{E(\psi)} \leq \val_{\Pi}(\psi) (1 - \lambda) m < s m $.
        Applying the multiplicative Chernoff bound (see \Cref{thm:mult-chernoff-2}), we get that
        \begin{align*}
           \prob{\card{E_1(\psi)} \geq s'm} & = \prob{\card{E_1(\psi)} \geq \frac{s'm}{\mu} \cdot \mu}
            \leq \lb \frac{e^{\frac{s'm}{\mu} - 1}}{\lb \frac{s'm}{\mu}\rb^{\frac{s'm}{\mu}} }\rb^{\mu} 
            \leq \exp\lb s'm \rb \lb \frac{ s}{s'} \rb^{s' m},
        \end{align*}
        where we used that $\mu \leq s m$ in the last inequality. Substituting the value of $s$ and $s'$, we get that:
        \begin{align*}
            \prob{\card{E_1(\psi)} \geq s'm} & 
            = \lb \frac{e C}{C' R^{\gamma- 1}} \rb^{s' m}
           =  \lb \frac{e C}{C' R^{\gamma- 1}} \rb^{\frac{C'}{R}\cdot \frac{R}{k}n} 
            \leq R^{-\frac{(\gamma - 1)(1 - \lambda)C'}{k} \cdot n},
        \end{align*}
        where we used that $R_0 \leq R$ and that $m = nR/k$.
        We compute the probability that there exists one assignment that satisfies more than an $s'$ fraction of the constraints using a union-bound over all $R^n$ assignments:
        \begin{align*}
           \prob{\bigvee_{\psi} \bar{\cE}_\psi} & \leq \sum_{\psi} \prob{\bar{\cE}_\psi} \leq R^n \cdot R^{-\frac{(\gamma - 1)(1 - \lambda)C'}{k} n} = \lb R^{n/k} \rb^{k - (\gamma - 1)(1- \lambda)C'}
        \end{align*}
        Substituting $C' = \frac{k(1 + \lambda)}{(\gamma - 1)(1 - \lambda)}$, then
        \begin{align*}
            \prob{\bigvee_{\psi} \bar{\cE}_\psi} & \leq R^{n (1 - (1 + \lambda))} = R^{-\lambda n} \leq R^{-1} \leq 0.01
        \end{align*}
        where we used that $n \geq 1$ and $R \geq 100^{1/\lambda}$.
        We finish the proof by computing the fraction of constraints that are satisfied by any assignment. Condition of $\cE_1, \cE_2$ (\Cref{thm:ClaimA}, \Cref{thm:ClaimB}) and on $\bigwedge_{\psi} \cE_{\psi}$, with probability at least $0.97$ we have that:
        \begin{align*}
           \val(\Pi') & \triangleq \max_{\psi} \frac{\card{E_2(\psi)}}{\card{E_2}} \leq \max_{\psi} \frac{\card{E_1(\psi)}}{\card{E_1} - \card{E_1 \bb E_2}} \leq \max_{\psi} \frac{s' m}{(1 - \lambda)m} = \frac{C'}{(1 - \lambda)R} = \frac{k(1 + \lambda)}{(\gamma - 1)(1-\lambda)^2R}. \qedhere
        \end{align*}
        \end{proof}

        \begin{proof}[Proof of \Cref{thm:bounded-degree-hardness-simple}]
            It follows a simple combination of \Cref{thm:4-CSP-Hardness} and \Cref{thm:4-CSP-bounded-hardness}.
            Fix $k \in \NN$. Let $\e > 0$ and $\lambda \in (0, 1)$ such that $(1 + \lambda)/(1-\lambda)^2 \leq 1 + \e$.
            By \Cref{thm:4-CSP-Hardness}, there is $(1 - \e, O(R^{-(k-2)}))$-hard $d$-regular $k$-CSP instance $\Pi$ over alphabet size $R$. For $k \geq 4$, we apply \Cref{thm:4-CSP-bounded-hardness} to obtain a $(1 - \e, \frac{k(1 + \e)}{(k - 3)R})$-hard $R$-degree bounded $k$-CSP instance $\Pi'$ with alphabet size $R$.
        \end{proof}

         \subsection{Conclusion and Open Questions}
            The main contribution of this paper is an improved hardness result for $k$-Dimensional Matching equal to $k/12$ for large values of $k$ and improves over the $O(k/\log(k))$-hardness from \cite{Hazan:2006:Complexity}. It uses an (arguably) clean approximation preserving gadget to encode satisfying assignments of $R$-degree bounded $k$-CSP over alphabet size $R$ into matchings in a $kR$-dimensional matching instance. We prove that $R$-degree bounded $k$-CSP over alphabet size $R$ are hard to approximate within a factor $\frac{k}{(k-3)R}$ using the randomized sparsification method from \cite{Lee:2024:Hardness}. The result then follows from combining these two facts.
            At a higher level, our result narrows the gap between approximability and hardness for $k$-Dimensional Matching from $O(\log(k))$ to a constant. Our result directly implies that $k$-Set Packing, $k$-Matroid Intersection, $k$-Matchoid, and $k$-Matroid Parity are hard to approximate within a factor of $k/12$. \\

            Closing this gap is an interesting direction for future research. We believe that our hardness result can be improved by understanding the tight approximability of CSPs with bounded degree $d$ and alphabet size $R$.
    One possible way is to better understand the bounded-alphabet-only case.
    For instance, the techniques from Theorem~\ref{thm:degree-reduction-general} show that if the best-known $O(\log R / R^{s-1})$-hardness holds for $s$-CSP with alphabet size $R$ with {\em almost perfect completeness}\footnote{In the completeness case, the normalized value of the instance is at least $1 - \eps$. It is already proved to be optimal without this restriction~\cite{khot2015approximating, lee2022characterization}.}, for any $s \geq 3$, then one can reduce the degree to $R$ with new soundness $\simeq \frac{s}{(s-2)R}$. Combined with our reduction (Theorem~\ref{thm:reduction-to-SP}) to $k$-Set Packing that set size $k = sR$, it implies a $\simeq (\frac{s-2}{s^2}) k$-hardness for $k$-Set Packing which and would improve over \Cref{thm:main-matching} with a stronger $k/8$-hardness by setting $s = 4$.
    Of course, there might be more direct ways to understand the approximability of degree-$d$ alphabet-$R$ CSPs, bypassing Theorem~\ref{thm:degree-reduction-general}. Similarly, for $k$-Set Packing, one might design a different gadget that bypasses Theorem~\ref{thm:reduction-to-SP}, which requires $d = R$.

        \printbibliography
        \appendix

        \section{Proof of \Cref{thm:main-matching}}
        \label{sec:simple-hardness-computation}
        \Main*
        \begin{proof}[Proof of \Cref{thm:main-matching}]
            We would like to prove that, for any $\e > 0$ and any $p \geq p_0(\e)$ approximating $p$-DM beyond a factor of $(12 + \e)/p$ is hard unless $\NP \subseteq \BPP$. The proof almost follows from the combination of \Cref{thm:bounded-degree-hardness-simple} and the reduction from \Cref{thm:reduction-to-SP}. But the reduction in \Cref{thm:reduction-to-SP} needs $p = kR$ for some prime $R$ and some integer $k \in \NN$. Circumventing this problem can be done using the existence of a close number of the form $kR$ such that $p/(kR) \leq 1+\e$ assuming that $p$ is large enough. \\
            
            Fix $\delta \in (0, 1)$. Using the Prime Number Theorem (about the density of primes) (see for instance \cite{Baker:2001:Difference}), for any $\e_2 > 0$ and large $p \geq p_0(\e_2, \delta)$ there exists a prime $R$ such that $(1 - \e_2) p \leq kR \leq p$. Observe that this $R$ can be found in polynomial time.
            Assuming that $p$ is large enough so that $R$ is large enough, we apply \Cref{thm:bounded-degree-hardness-simple} to obtain a $R$-degree bounded $k$-CSP instance $\Pi$ over an alphabet of size $R$ with gap $[1 - \delta, \frac{k}{(k-3)R} (1+\delta)]$. We then use \Cref{thm:reduction-to-SP} to get (in polynomial time) a $kR$-Dimensional Matching instance $G = (V, E)$ such that $\card{\cC(\psi^\ast)} = \card{M^\ast}$. We transform this $kR$-DM instance into a $p$-DM instance $G'$ by adding dummy nodes. More precisely, we extend the vertex set by adding disjoint sets $D_1, \ldots, D_{p - kR}$ each containing $\card{E}$ vertices. The edges of $G'$ are obtained as follows: we order the edges in $G$ and for $e_i \in E$ we add the $\nth{i}$ vertex from each $D_j$ with $j \in [p - kR]$. So an edge in $G'$ consists of some $e \in E$ and $p - kR$ dummy vertices. Note that each dummy vertex is incident to only one edge. It is fairly easy to verify that $G'$ is $p$-partite (since $G$ is $k$-partite), that the matching size is preserved, and that this construction takes polynomial time in $p$ and $\card{E} = O(R^k)$ since $\Pi$ is $R$-degree bounded with alphabet size $R$. \\

            Suppose by contradiction that there is a $\frac{(k-3)p(1- \e)}{k^2}$-approximation algorithm for $p$-Dimensional Matching. The following computation proves that we would be able to distinguish the two CSP-cases contradiction \Cref{thm:4-CSP-bounded-hardness}. Indeed, suppose first that $\val(\Pi) \geq 1 - \delta$, then the algorithm returns on $G'$ a matching of size:
            \begin{align*}
                \card{M} \geq \frac{k^2}{(k-3)p(1-\e)} \card{M^\ast} \geq \frac{k(1 - \delta)(1- \e_2)}{(k-3)R(1 - \e)} \card{E}  > \frac{k}{(k-3)R}(1 + \delta) \card{E},
            \end{align*}
            where we used that $(1 - \e_2)p \leq k R$ and that $\delta$ and $\e_2$ can be chosen as arbitrarily small constant depending on $\e$.
            Alternatively, whenever $\val(\Pi) \leq \frac{k}{(k-3)R}(1 + \delta)$, the algorithm would return a matching of size: $\card{M} \leq \frac{k(1 +\delta)}{(k-3)R} \card{E}$. In particular, the algorithm would be able to distinguish the completeness and soundness case. By setting $k = 6$, unless $\NP \subseteq \BPP$, for any $\e > 0$, there is no polynomial time algorithm that approximates $p$-Dimensional Matching with a factor of $12/(p\cdot  (1 - \e))$.
        \end{proof}

        \section{Proof of claims}
        \label{sec:proba-computations}

            \ClaimA*
            \begin{proof}[Proof of \Cref{thm:ClaimA}]
                The expected value of $\card{E_1(\psi)}$ is equal to:
                \begin{align*}
                \esp{ \card{E_1(\psi)} } & = p \cdot \card{E(\psi)} = p \card{E} \val_{\Pi}(\psi) = (1 - \lambda) m \val_{\Pi}(\psi), 
                \end{align*}
                as every constraint gets added to $E_1$ with probability $p$.
                On the other hand, $ \sigma^2 \triangleq \var\ld \card{E_1(\psi)} \rd  \leq p (1 - p) \card{E} \val_{\Pi}(\psi) \leq (1 - \lambda) m \val_{\Pi}(\psi) \leq m$.
                Applying \Cref{thm:Cantelli}, we have that:
                \begin{align*}
                    \prob{\card{E_1(\psi)} - \expect{\ld \card{E_1(\psi)} \rd} \leq  - \lambda m} & \leq \frac{\sigma^2}{\sigma^2 + (\lambda m )^2} \leq\frac{m}{ m + (\lambda m)^2} = \frac{1}{1 + \lambda^2 m} \leq 0.01,
                \end{align*}
                where we used that $m \geq R \geq 100 \lambda^{-2}$.
                Thus, with probability at least $0.99$, we have that
                \begin{align*}
                    \card{E_1(\psi)} \geq  \expect{\ld \card{E_1(\psi)} \rd} - \lambda m & \geq (1 - \lambda) m \val_{\Pi}(\psi) - \lambda m \geq m \lb \val_{\Pi}(\psi) - 2\lambda \rb. \qedhere
                \end{align*}
            \end{proof}

            \ClaimB*
            \begin{proof}[Proof of \Cref{thm:ClaimB}]
            The expected number of deletions is equal to:
            \begin{align*}
            \expect{\ld \card{E_1 \bb E_2} \rd} & \leq \sum_{i = 1}^k \sum_{u \in V_i} \expect{\ld \deg_{G_1}(u) - \min \lc R, \deg_{G_1}(u) \rc \rd}.
            \end{align*}
            Let $X_e$ be the Bernoulli random variable equal to $1$ if $e \in E_1$, that is $\prob{X_e = 1} = p $ and observe that $\deg_{G_1}(a) = \sum_{e \in \delta(a)} X_e$. Thus, $\deg_{G_1}(a)$ is the sum of $d$ Bernoulli random variables with mean equal to $(1 - \lambda)R$. We can therefore apply \Cref{thm:clip-bound} to obtain:
            \begin{align*}
                \expect{\ld \deg_{G_1}(u) - \min \lc R, \deg_{G_1}(u) \rc \rd} & = \frac{\lb d p \rb^2}{\lb R - d p \rb^2} = \frac{( 1- \lambda)^2}{\lambda^2} \leq \lambda^{-2}.
            \end{align*}
            Combining the previous equations, we then have that
            \begin{align}
                \expect{\ld \card{E_1 \bb E_2} \rd}  &  \leq \lambda^{-2} \lb \sum_{i = 1}^k \card{V_i}\rb = \frac{k m}{R} \cdot \lambda^{-2} \leq 0.01 m \lambda,
            \label{eq:number-deletions}
            \end{align}
            where we used that $R \geq R_0  \geq k \cdot 100\lambda^{-3}$. We conclude using Markov's inequality: $\prob{\card{E_1 \bb E_2} \geq \lambda m} \leq \frac{\expect{\ld \card{E_1 \bb E_2} \rd }}{\lambda m} \leq 0.01.$
            \end{proof}

            \ClaimC*
            \begin{proof}[Proof of \Cref{thm:ClaimC}]
                The proof follows from Chebyshev's inequality applied to $E_1 = \sum_{e \in E} X_e$ where $X_e$ is a Bernoulli random variable equal to $1$ if $e \in E_1$ and 0 otherwise.
                Then,
                \begin{align*}
                    \prob{\bar{\cE}_3} = \prob{\card{E_1 - \esp{\card{E_1}}} \geq \lambda m}\leq \frac{\var\lb E_1 \rb}{\lambda^2 m^2} \leq \frac{m}{\lambda^2 m^2} \leq \frac{1}{\lambda^2 m}\leq 0.01,
                \end{align*}
                where the last inequality uses that $m \geq R \geq R_0 \geq 100 \lambda^{-2}$.
            \end{proof}

        \section{Probability Theorems}
        \begin{theorem}[Cantelli's inequality]
            \label{thm:Cantelli}
            Let $X$ be a random variable with finite variance $\sigma^2$ (and thus finite expected value $\mu$). Then, for any real number $\alpha > 0$:
            \begin{align*}
               \prob{ X - \mu \leq - \alpha } & \leq \frac{\sigma^2}{\sigma^2 + \alpha^2}.
            \end{align*}
        \end{theorem}
        \begin{theorem}[Multiplicative Chernoff Bound]
            \label{thm:mult-chernoff-2}
            Let $X_1, \ldots, X_m$ be i.i.d Bernoulli random variables. Let $S = \sum_{i = 1}^m X_i$ denote their sum and let $\mu = \expect{\ld S \rd}$. Then, for any $\delta > 0 $, we have that
            \begin{align*}
               \prob{S > (1+\delta)\mu} & < \lb \frac{e^{\delta}}{(1+\delta)^{1+\delta}} \rb^\mu.
            \end{align*}
        \end{theorem}
        \begin{theorem}[Theorem 6 \cite{Lee:2024:Hardness}]
            \label{thm:clip-bound}
            Let $X_1, \ldots, X_m$ be i.i.d Bernoulli random variables with mean at most $\mu$ and let $S = \sum_{i \in [m]} X_i$. Then, for any integer $\tau > \mu m$, we have that
            \begin{align*}
                \esp{S - \min\{S, \tau\}} \leq \lb \frac{\mu m }{ \tau - \mu m} \rb^2.
            \end{align*}
        \end{theorem}

\end{document}